 \documentclass[]{imsart}
\RequirePackage{amsthm,amsmath,amsfonts,amssymb}
 
\RequirePackage{natbib}
\RequirePackage[colorlinks,citecolor=blue,urlcolor=blue]{hyperref}
\RequirePackage{graphicx}
\RequirePackage{algorithm2e}
\RequirePackage{url}
\startlocaldefs

\newcommand{\change}[1]{\textcolor{black}{#1}}
\theoremstyle{plain}
\newtheorem{axiom}{Assumption}

\newtheorem{theorem}{Theorem}[section]
\newtheorem{lemma}[theorem]{Lemma}
\theoremstyle{definition}

\theoremstyle{remark}
\newtheorem{case}{Remark}
\endlocaldefs

\begin{document}
\begin{frontmatter}
\title{Box Confidence Depth: simulation-based inference with hyper-rectangles}
\runtitle{Box Confidence Depth}

\begin{aug}
\author[A]{\fnms{Elena}~\snm{Bortolato}\ead[label=e1]{elena.bortolato@bse.eu}\orcid{0000-0001-6744-3090}},
\author[B]{\fnms{Laura}~\snm{Ventura}\ead[label=e2]{ventura@stat.unipd.it}\orcid{0000-0001-7322-998X}}

\address[A]{Department of Economics, Universitat Pompeu Fabra \& Data Science Center,  Barcelona School of Economics \printead[presep={,\ }]{e1}}

\address[B]{Department of Statistical Sciences,
University or Padova\printead[presep={,\ }]{e2}}
\runauthor{E. Bortolato et al.}
\end{aug}

\begin{abstract}
	This work presents a novel simulation-based approach for constructing confidence regions in  parametric models, which is  particularly suited for  generative models and situations where limited data and  conventional asymptotic approximations fail to provide accurate results.
	The method leverages the concept of data depth  and depends on creating random hyper-rectangles, i.e. boxes, in the sample space  generated through simulations from the model, varying the input parameters.
	A probabilistic acceptance rule allows to retrieve a Depth-Confidence Distribution for the model parameters from which point estimators as well as calibrated confidence sets can be read-off.
	The method is designed to address cases where both the parameters and test statistics are multivariate.
\end{abstract}

\begin{keyword}[class=MSC]
\kwd[Monte Carlo methods ]{65C05}
\kwd[; Tolerance and confidence regions ]{62F25}
\kwd[; Order statistics; Empirical distribution functions ]{62G30}
\end{keyword}

\begin{keyword}
\kwd{Depth functions}
\kwd{Simulation-Based inference}
\end{keyword}

\end{frontmatter}

\section{Introduction}
In many scientific domains, researchers face the challenge of evaluating complex statistical models in which the likelihood function is either computationally intractable or prohibitively expensive to calculate. This has led to the development and increasing popularity of likelihood-free inference methods, which offer powerful alternatives for parameter estimation and model comparison.
These methodologies leverage simulations, enabling  inference through the comparison of observed data with simulated outcomes generated from the model under various parameter settings.
In Bayesian inference, these include
Approximate Bayesian Computation  \citep{ rubin1984bayesianly,  pritchard1999ABC, sisson2018handbook},
Bayesian Synthetic Likelihood   \citep{wood2010statistical,price2018bayesian}, 
{Neural Likelihood and Posterior Estimation}  \citep{rezende2015variational,  papamakarios2019sequential}. 
In the frequentist setting,  after the foundational work of \cite{gourieroux1993indirect}, only recent years have seen advancements in likelihood-free inference \citep{ masserano2022simulation, xie2022repro, dalmasso2024likelihood}.


This study focuses on frequentist inference, 
targeting the construction of calibrated confidence intervals and regions across simulation-based models and non-standard regularity conditions. 
The proposed approach  provides a unified strategy for inference that seamlessly accommodates both univariate and multivariate parameters. This is achieved by means of a depth function   \citep{liu1990notion}, that allows defining nested confidence sets across all confidence levels,  offering researchers a comprehensive visualization of parametric uncertainty. 
A significant aspect of the proposed methodology is its ability to operate without requiring data to be necessarily reduced to a scalar summary statistic, \change{as it  is typically done in the frequentist framework}. Raw data can be utilized directly, enhancing the flexibility and automation of the inference process. \change{Similarly,  inference from diverse test statistics, linked to model-specific information, can be combined in a natural manner}. 
As a   byproduct of the procedure, the method also yields consistent   point estimators for model parameters.  

The rest of the paper is organized as follows. Section 2 reviews recent developments in simulation-based inference. Section 3 outlines the sampling methodology used to build the Confidence Depth, discusses its theoretical underpinnings, and addresses some  computational aspects, \change{with particular emphasis on challenges and remedies related to the curse of dimensionality problem}. Section 4 discusses various examples from either classical models such as Generalized Linear Models (GLMs), as well as models from the field of Likelihood Free Inference (LFI) and reports simulation studies. A discussion is provided in Section 5.

\section{Simulation based inference}
Consider a parametric model $p(y|\theta)$, with $\theta$  a finite-dimensional parameter. 
We denote with  {$y^\text{obs}$} the observed data, of size $n$,  with $ {t}:\mathbb{R}^n\rightarrow \mathbb{R}^d$ a collection of summary statistics of $d\leq n$ components, with
$ {t^\text{obs}}=t(y^\text{obs})$ the observed summary statistics. %

The key idea of Simulation Based Inference (SBI)  is that  inference can  rely on simulations from   the same process responsible for producing observed data.  Once pseudo-observations are generated from the model across various parameters values, the plausibility of the parameter used in the simulation can be assessed, based on comparison with the original data $y^\text{obs}$.

The most popular method for SBI in Bayesian inference is 
Approximate Bayesian Computation (ABC), introduced by  \cite{rubin1984bayesianly} and further developed by \cite{pritchard1999ABC}. ABC aims to generate  datasets  that mimic the observed sample using as proposals for $\theta$ draws from the prior distribution. Parameter values that generate synthetic observations closely matching the real observation, up to a certain tolerance $\varepsilon$, i.e. $d(y, y^\text{obs})<\varepsilon$, are retained. 
The   distance or divergence $d(\cdot, \cdot)$ between pseudo and actual data  is generally assessed  on a set of summary statistics that are   informative for the model.
Intuitively, if the synthetic data match the observed data, the model parameters used in the simulations are plausible for the model under consideration and in turns they are associated to  higher  likelihood function.  Several enhancements to the basic ABC algorithm have been proposed over time, see \cite{marjoram2003markov, marin2012approximate, del2015sequential, frazier2018asymptotic, bernton2019approximate,   rotiroti2024approximate} and references therein. 
The approximation in ABC is considered  non parametric, as the shape of the likelihood and the posterior is not specified but obtained by rejection Monte Carlo.
Parametric approximations of likelihood functions (and posteriors) in simulation-based settings have seen significant advancements in recent years, largely due to the growing influence of Machine Learning and Deep Learning techniques.  Two prominent approaches are  Bayesian Synthetic likelihood \citep{wood2010statistical,  price2018bayesian,    frazier2023bayesian}, which employs conditional density estimators based on a multivariate Gaussian model and  the family of 
Neural Posterior Estimation methods  \citep{rezende2015variational, papamakarios2019sequential} employing more flexible  conditional density estimators, as normalizing flows which better suited for high-dimensional data and complex models.
Machine Learning methods  have been heavily employed in Neural Ratio Estimation (NRE)  \citep{hermans2020likelihood, thomas2022likelihood} that estimates the ratio between the likelihood  $p(y|\theta)$ and data marginal $p(y)$, that is $r(y, \theta) = p(y|\theta)/p(y)$,  by training a  classifier  to distinguish  datasets generated from the 
conditional and the  marginal model. 

In the frequentist paradigm,
ratio estimation was  adopted by \cite{dalmasso2024likelihood}
to approximate the likelihood ratio statistic. In particular, once the  quantity $r(y, \theta)$ is approximated by means of the classifier trained  on the conditional model and  on models simulated using a reference distribution for the parameter of interest,  the empirical quantiles of level $\alpha$     are used to build confidence sets. Recently, \cite{kuchibhotla2024hulc} developed a methodology for constructing confidence intervals and sets with bounded coverage errors by utilizing   data subsampling. Nevertheless, the strategy is not purely likelihood-free, as it relies on Maximum Likelihood estimation, similarly to the Bootstrap approach \citep{efron1979computers, efron2003second}.

\section{Box-Confidence Depth}
Assume that it is possible to generate data from the parametric model $p(y|\theta)$, with $\theta\in \Theta \subseteq \mathbb{R}^p$. Let   $ {\theta_0}$ be the true value of $\theta$.   We assume that the model is correctly specified, so that    $ {p(y|\theta_0)}$ corresponds to the  true data generating process. Let $\pi(\theta)$ be a proposal distribution for the unknown parameter. Here, the proposal distribution is assumed to be uniform in a bounded subset of the parameter space $\Theta^b \subset \Theta$.    \textcolor{black}{The boundedness of $\Theta^b$, even when the parameter space is naturally unbounded, is technically necessary  to ensure computational feasibility, and guidelines to choose $\Theta^b$ without compromising the validity of the inferential procedure are discussed below.}

The proposed method consists in drawing $\theta^*$ from $\pi(\theta)$ and, for each $\theta^*$, generate two pseudo-samples from the model ${p(y|\theta^*)}$, denoted as $y^{*1}$ and $y^{*2}$. Summary statistics $t^{*1}$ and $t^{*2}$, each of dimension $d$, are then computed from $y^{*1}$ and $y^{*2}$, respectively. The proposal $\theta^*$ is accepted if the observed summary statistic $t^{\text{obs}}$, computed from the actual observed data $y^{\text{obs}}$, falls within a region defined by $t^{*1}$ and $t^{*2}$. This region can be conceptualized as a $d-$dimensional hyper-rectangle, called Box and denoted as $\mathcal{B}_t^*$,  in the space of summary statistics, with $t^{*1}=(t_{1}^{*1}, \ldots, t_{d}^{*1})$ and $t^{*2}=(t_{1}^{*2}, \ldots, t_{d}^{*2})$ defining its edges, i.e.
$$\mathcal{B}_t^*=\times_{j=1}^{d}[t_{j}^{*
	(1)},t_{j}^{*(2)}],$$ where $t_{j}^{*(1)}$ and $t_{j}^{*(2)}$ are the order statistics along the $j$-th coordinate ($j=1, \ldots, d)$. Equivalently, the parameter $\theta^*$ is accepted if $t_1^{*(1)} <t_1^{\text{obs}}< t_1^{*(2)} ,\:
t_2^{*(1)}< t_2^{\text{obs}} <t_2^{*(2)},\: 
\ldots, \:
t_d^{*(1)} < t_d^{\text{obs}} < t_d^{*(2)}.    
$  Figure \ref{fig:accrej} illustrates this concept in dimension $d=2$ and the algorithm is outlined in Algorithm \ref{algo0}.

\begin{algorithm}[h!]
	\caption{Accept-Reject Box-CD}
	\label{algo0}
	\KwIn{Proposal distribution $\pi(\theta)$, number of iterations $R$, summary statistic $t(\cdot)$, observed statistic $t^{\text{obs}} = t(y^{\text{obs}})$}
	\KwOut{Accepted samples $\theta^*$}
	
	\For{$j \gets 1$ \textbf{to} $R$}{
		Sample $\theta^*_j \sim \pi(\theta)$\;
		Sample $y_j^{*1}, y_j^{*2} \sim p(y|\theta_j^*)$\;
		Compute $t_j^{*1} = t(y_j^{*1})$, \quad $t_j^{*2} = t(y_j^{*2})$\;
		\If{$t^{\text{obs}} \in \mathcal{B}_t^* $}{
			Accept $\theta_j^*$\;
		}
	}
	\Return Accepted samples $\theta^*$
\end{algorithm}
\begin{figure}[h!]
	\includegraphics[width=1\linewidth]{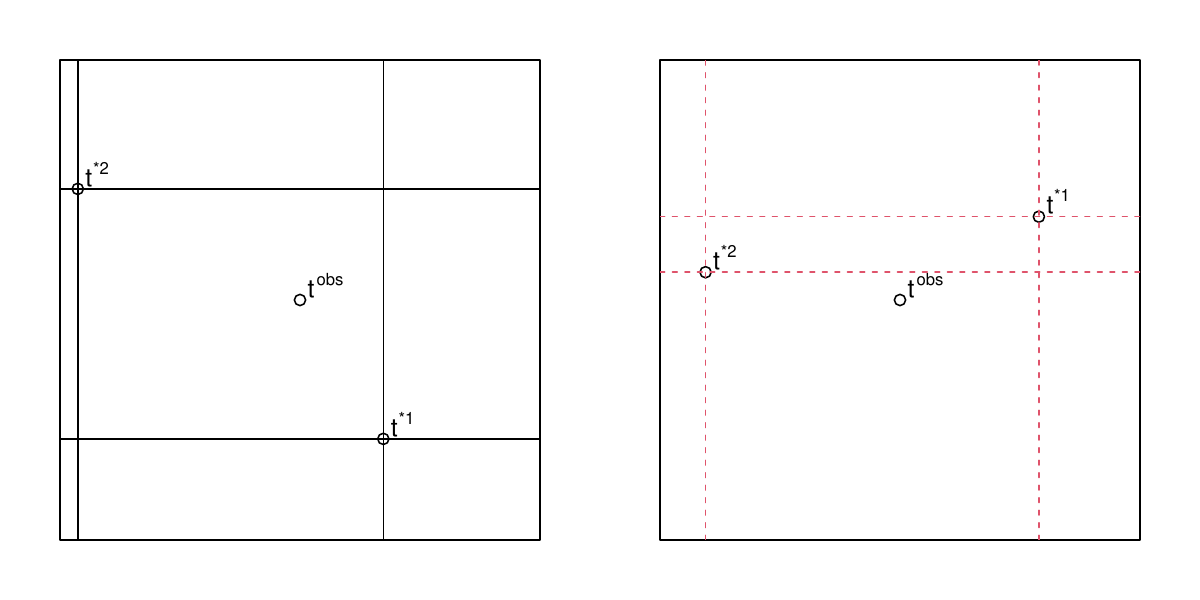}
	\caption{   
		Two examples of summary statistics, \( t^{*1} = (t^{*1}_1, t^{*1}_2) \) and \( t^{*2} = (t^{*2}_1, t^{*2}_2) \) computed on simulated pseudo-samples.  {Left}: the proposal parameter would be accepted as the observed  \( t^{\text{obs}} \) lies within the Box. {Right}: the proposal is rejected as \( t^{\text{obs}} \) falls outside this Box. }\label{fig:accrej}
\end{figure}

The procedure described seeks to learn a data-dependent distribution over the parameter space non-parametrically, utilizing a rejection algorithm, as in ABC, instead of assuming a closed-form model for summary statistics.
{However, it deeply differs from ABC as the inclusion criterion is not a distance function or a divergence. In particular, distances or divergences employed in ABC satisfy the identity of indiscernibles  property, i.e. $d(a,b)=0 \iff a=b$. Thinking about the proposed criterion illustrated in Algorithm \ref{algo0} as a discrepancy, it still may be zero even if the simulated data don't perfectly align with the observed sample. Conversely, at least ideally in ABC, as the tolerance or threshold parameter narrows, the pseudo-data must precisely match the observed data. The concept of matching simulations to observed data aligns with the idea of conditioning, which is a fundamental aspect of ABC. In contrast, the notion of ordering by using a series of inequalities in the sample space conforms to frequentist reasoning.}

Note that, since establishing a meaningful ordering in the sample space becomes challenging in presence of multivariate data and summary statistics, the procedure in practice utilizes a measure of centrality of observed data with respect to simulated data, which corresponds to an ordering   from the center outwards.

Proposals $\theta^*$  associated with a high centrality of the observed sample lead to frequent acceptance. This results in an empirical Monte Carlo-based measure of confidence and an ordering within the parameter space. 
The accepted $\theta^*$ are distributed as $$\mathcal{CD}^{\text{box}}(\theta): \Theta\mapsto \change{ \mathbb{R}^+},$$ called Box-Confidence Depth, which  assigns a measure of centrality or "depth"  to each parameter, with  higher values indicating that the parameter is more central or representative of the observed data.

 {As a final remark, observe that the procedure can be generalized to consider more than two replicas of the data. This extension is discussed in Section 3.5, after detailing the method's properties  and the main results.}

\subsection{Scalar-scalar case}
To formalize the properties of the function $\mathcal{CD}^{\text{box}}(\theta)$ in relation to confidence intervals and frequentist tests, it is useful to initially consider the scenario where $\theta \in\change{ \mathbb{R}}$ and $d=1$. In this case,  the Box reduces to an interval with  endpoints $t^{*1}=t(y^*_1)$ and $t^{*2}=t(y^*_2)$ and  the  proposed $\theta^*$ is accepted if and only if
$ 
t(y^\text{obs}) \in [t^{*(1)} , t^{*(2)}] 
$.   

\begin{axiom}
	The statistic $t: \mathcal{Y} \mapsto \mathcal{T} \in \mathbb{R}$ is one-dimensional
	and $0<Var(t(y)|\theta)<\infty$ for $\pi$-almost all $\theta$. \label{ASSstatistic}
\end{axiom}
The assumption that $Var(t(y)|\theta)>0$ for $\pi$-almost all $\theta$
ensures that the intervals of the form $[t^{(1)} , t^{(2)}]$ have positive probability of being
non-empty. 

\begin{axiom} \label{support}
	The support $\Theta^b$ is chosen such that $$\sup_{\theta \in \Theta^b}  F_t(t^{\text{obs}}|\theta)
	>1-b \text{ and }
	\inf_{\theta \in \Theta^b}  F_t(t^{\text{obs}}|\theta)< b,$$ where  $F_t(t^{\text{obs}}|\theta)$  is the cumulative \change{distribution} function  of $t(y)$,  computed \change{at} the value of the observed summary statistic $t^{\text{obs}}$ and  with \textcolor{black}{$b$} of the same order of the machine tolerance.
\end{axiom}
\textcolor{black}{The set $\Theta_b$ serves a purely technical purpose: it  guarantees that all parameter values with acceptance probability exceeding machine precision are included.}
\begin{lemma}\label{lemmatarget}
	For a scalar parameter $\theta$, under Assumption \ref{ASSstatistic},
	the Box-Confidence Depth   is 
	$$ \mathcal{CD}^{\text{box}}(\theta) \propto F_t(t^\text{obs}|\theta) [1- F_t(t^\text{obs}|\theta)].  $$
\end{lemma}
\begin{proof}
	Let $\theta\in \Theta$, and consider a pair of statistics following the pushed-forward distribution induced by the summary statistic $t$ applied to $y\sim p(y|\theta)$, i.e. $(t^{*1},\; t^{*2}) \overset{iid}{\sim} t_\# p(y|\theta)$. 
	One can compute the probability of acceptance of $\theta$  as follows:
	\begin{align*}
		&\Pr(t^{(1)} \leq t^{\text{obs}} < t^{(2)}|\theta)= 
		\Pr(t^{*1} \leq t^{\text{obs}} < t^{*2}|\theta)  
		+ \Pr(t^{*1} > t^{\text{obs}} \geq t^{*2}|\theta)\\
		&= \Pr(t^{*1} \leq t^{\text{obs}},   t^{\text{obs}} < t^{*2}|\theta)+
		\Pr(t^{*1} > t^{\text{obs}}, t^{\text{obs}} \geq t^{*2}|\theta)\\
		&= F_t(t^{\text{obs}}|\theta) [1-F_t(t^{\text{obs}}|\theta)] +F_t(t^{\text{obs}}|\theta) [1-F_t(t^{\text{obs}}|\theta)]\\
		&=2 F_t(t^{\text{obs}}|\theta) [1-F_t(t^{\text{obs}}|\theta)]\\
		&\propto  F_t(t^{\text{obs}}|\theta ) [1-F_t(t^{\text{obs}}|\theta)].
	\end{align*} Under Assumption \ref{support} we obtain the target distribution by a usual rejection sampling argument. 
\end{proof}

\begin{lemma}\label{lemma:pval}
	Under Assumption \ref{ASSstatistic}:\\
i)\:$F_t(t^{\text{obs}}|\theta)$ as a function of $\theta$ is   a one-sided $p$-value function,\\
ii)\:
		$F_t(t^{\text{obs}}|\theta)$ is a Confidence Distribution (CD) when  $F_t(t^{\text{obs}}|\theta)$ is stochastically increasing in $\theta$.
\end{lemma}
\begin{proof}
	Statement $i)$ follows immediately from the definition of $F_t$ in Lemma \ref{lemmatarget}. 
	For $ii)$, in general,
	a function $H(y, \theta)$ on $\mathcal{Y} \times \Theta \rightarrow  [0, 1]$ is  a  CD for  $\theta$ if {\citep[see e.g.][]{xie2013confidence}}: a)
	for each given $y \in \mathcal{Y}, H(\cdot)$ is a cumulative distribution function on $\Theta$; b)
	at  $\theta=\theta_0$, $H(y^{\text{obs}},\theta_0)$, as a function of the sample $y^{\text{obs}}$, follows a $\text{Uniform}[0, 1]$ distribution.     By construction $F_t(t^{\text{obs}}|\theta_0)\in[0,1]$, furthermore, if $t$ is stochastically increasing in $\theta$, then $F_t(t^{\text{obs}}|\theta')< F_t(t^{\text{obs}}|\theta'')$ for $\theta''>\theta'.$  By properties of the $p$-value function and for  $u\in [0,1]$ 
	$Pr(F_t(t^{\text{obs}}|\theta)<u)= Pr(
	t^{\text{obs}}<F_t^{-1}(u|\theta))$  is constant when $t^{\text{obs}}$  is drawn from $F_t(\cdot|\theta)$.
\end{proof}

\textcolor{black}{\begin{case}\label{medianu}
	Let $\hat \theta$ be the maximizer of   $\mathcal{CD}^{\text{box}}(\theta) $. Then,  under Assumption \ref{ASSstatistic},    $\hat \theta$ is median unbiased, i.e. \begin{equation}\label{eq1}
		Pr_{\theta_0}(\hat \theta \leq\theta_0)=1/2.
	\end{equation}
	Indeed, by definition  $\hat \theta=\underset{\theta}{\arg\max} \: F_t(t^{\text{obs}}|\theta)[1-F_t(t^{\text{obs}}|\theta)]$ and, since $F_t(t^{\text{obs}}|\theta)\in [0,1]$, the expression is maximum when the function $x(1-x)$ is maximum with $x$ in $[0,1]$, which is  $F_t(t^{\text{obs}}|\hat\theta)=0.5$. Then, applying $F_t$ to both sides of the inequality of Equation (\ref{eq1}) it follows that
	$Pr_{\theta_0}(F_t(t^{\text{obs}}|\hat\theta)<F_t(t^{\text{obs}}|\theta_0))=$ $Pr_{\theta_0}(0.5< F_t(t^{\text{obs}}|\theta_0))=1/2.$ \end{case}}
\begin{case}\label{remark_med}
Median unbiasedness is a desired property as it guarantees  consistency of the estimator \citep{schweder2016confidence}. Additionally, this property is preserved  also for any one-to-one reparametrizations \citep{kenne2017median, kuchibhotla2024hulc}.
\end{case}

\begin{case}\label{remark_diff_conf}
	Note that differently from a Confidence Distribution, the shape of   $\mathcal{CD}^{\text{box}}(\theta) $
	does not assume that $t$ is  stochastically ordered in $\theta$. In particular, if $F_t(t^{\text{obs}}|\theta)$ is not monotone in $\theta$, the function  $\mathcal{CD}^{\text{box}}(\theta) $ can be multimodal. Figure \ref{fig:f1mf} illustrates 
	this possibility.
\end{case}

\begin{case}
	Note that if the proposal $\pi(\theta)$ is centered on the confidence median, and the function $\mathcal{CD}^{\text{box}}(\theta) $ is symmetric, then the expected acceptance probability is 1/4.   This can be regarded as a practical guideline to tune the proposal.
\end{case}
\begin{figure}[h!]
	\centering \includegraphics[width=0.8\linewidth]{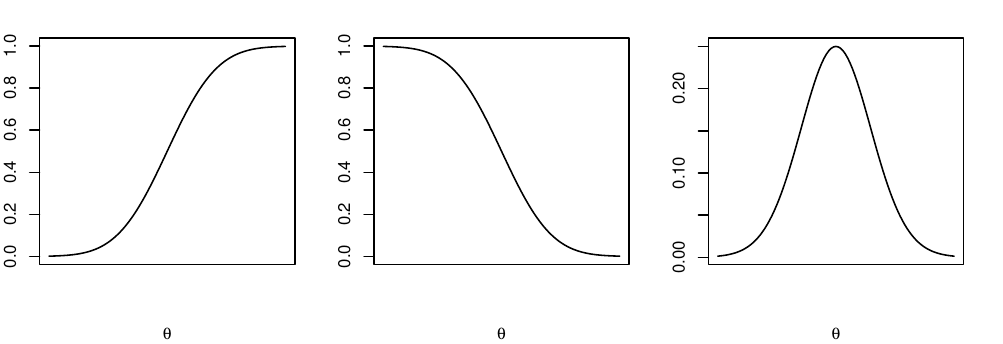}
	\includegraphics[width=0.8\linewidth]{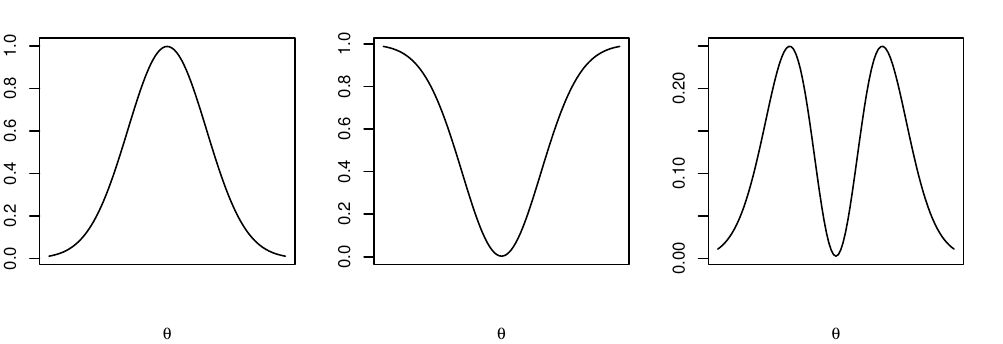}
	\caption{Top panels:
		an instance of a monotone   $p$-value function $F_t(t^{\text{obs}}|\theta)$
		(left), $1-F_t(t^{\text{obs}}|\theta)$ (center), and their product (right). Bottom panels: a non a monotone  $p$-value function where the resulting Confidence-Depth is multimodal.}
	\label{fig:f1mf}
\end{figure}

\subsection{Relation to confidence intervals} Let us write $F_t(\theta)$ shortly for  $F_t(t^{\text{obs}}|\theta )$.
Define  an equi-tailed confidence interval of size $1-\alpha$  as $C_{1-\alpha}=\{\theta| F_t(\theta)>\alpha/2  \text{\;and\:} F_t(\theta)<(1-\alpha/2)\}$.
Denote by $Q_Z(\cdot)$ the  function s.t.   for $p\in (0,1)$ and $Z,\zeta\in (0,1)$,
$Q_Z(p)=\zeta$ if $\int  1_{Z<\zeta  } dZ=p$. 

\begin{theorem}    \label{lemmaconfint}
	For a scalar parameter  $\theta$    and   $\alpha \in (0,1)$  the following relation holds: 
	$$ \mathcal{CD}^{\text{box}}(\theta)\geq Q_{\mathcal{CD}^{\text{box}}(\theta)}(\alpha) \Leftrightarrow  \theta \in C_{1-\alpha}.
	$$
\end{theorem}

\begin{proof}
	Consider the piecewise linear function $g(F_t)=-|F_t -0.5|\in [-0.5,0]$. From its definition, the set $C_{1-\alpha}$ can be written as   $\{\theta | $
	$ g(F_t)\geq-0.5+\alpha/2\}$. Since the image of $g$ is  $[-0.5,0]$ and the function is linear in $F_t$, if   ${g(F_t)}\geq-0.5+\alpha/2$ then $ {g(F_t)}\geq Q_g(\alpha)$.
	Applying the monotone transformation $h(F_t)= 2 g(F_t) \text{sign}(F_t -0.5) +F_t= 2 (0.5-F_t ) +F_t=F_t(1-F_t)$
	which is order preserving, it directly follows
	that $Q_{\mathcal{CD}^{\text{box}}(\theta)}(\alpha)= {g(F_t)}$. This concludes the proof as $C_{1-\alpha}$ can be written as
	$\{\theta | $
	$ \mathcal{CD}^{\text{box}}(\theta)\geq Q_{\mathcal{CD}^{\text{box}}(\theta)}(\alpha) \}$.
\end{proof}

Theorem \ref{lemmaconfint}  outlines how to define confidence intervals via the Box-CD method.
Indeed, the accepted  values from  Algorithm \ref{algo0} are draws from $\theta^*\sim \text{Bernoulli}(\mathcal{CD}^{\text{box}}(\theta))$.
Thus, it is sufficient to obtain a continuous approximation of the function $\mathcal{CD}^{\text{box}}(\theta)$. Specifically, any Machine Learning classification algorithm that outputs classification probabilities can be trained using proposals drawn from $\pi(\theta)$ as inputs, while the acceptance rule's outcomes (0 or 1) as labels.
Alternatively, the same task can be obtained by   density estimation starting from  the values $\theta^*$ accepted from Algorithm \ref{algo0}. From the parametric approximation (or density estimation) of the function $\mathcal{CD}^{\text{box}}(\theta)$, the quantiles can be obtained.


\subsection{Mutivariate case: center-outward ordering}

In classical statistics, when   multiple   statistics  are collected ($d>1$), assessing the $p$-value of a precise null hypothesis involves computing the tail area  probability of an event in  dimension $d$. This process is complex and involves several considerations, particularly due to the dependence of the statistics used. In practice, the joint distribution is often only well-defined for Gaussian distributed test   statistics, limiting the applicability to other distributions.
It is generally preferred to reduce the information in a one dimensional   statistic, even  when the inference is on a parameter vector $(p>1)$, such as  the Likelihood Ratio test (LR),    since
in contrast to univariate data, multivariate data lacks a natural method for ordering.

On the other side, to address the problem of ordering in multidimensional settings,  researchers have   developed various techniques  leveraging \textit{Data Depth} concepts.
Data Depth  (DD) functions  provide a measure of centrality within multivariate sample spaces quantifying how  deep a point is relative to a multivariate probability distribution or data cloud. This centrality measure allows for a center-outward ordering of points in any dimension to ultimately delineate nested
central regions. For example, the Simplicial Depth (SD) method introduced by \cite{liu1990notion} determines the depth of a point by evaluating its presence within all combinations of simplex formed by the data points.
When examining the univariate counterpart of the SD, i.e. two independent observations drawn from a univariate cumulative distribution function, the SD is reduced to the form
$\text{SD}_1(x)=2F(x)[1 - F(x)]$, and the point that maximises $\text{SD}_1(x)$ corresponds to the median of the population.
Note that the definition resembles that of the  Box-CD function $\mathcal{CD}^{\text{box}}(\theta)$ in the scalar case.
Another well known DD function is the Tukey's Depth (or Half-space Depth. HD). In one dimension it is used   as the $p$-value for bilateral tests: $$ \text{HD}_1=2\min\{Pr_\theta(Y\leq y^{\text{obs}} ),Pr(Y\geq y^{\text{obs}})\}.$$ The HD function in the multivariate case requires the definition of a convex hull, which is the intersection of all halfspaces containing all sample points. The level sets of the HD are defined as the intersections of  halfspaces containing $k<n$ sample points. 

\cite{liu2022datadepth} consider the concept of DD to define Confidence Distributions for multivariate parameters, called depth CDs, by ranking parameter values instead of data points. They propose to use the distribution of non-parametric bootstrap  estimates to recover an approximate depth CD, motivated by the fact that  algorithms for reconstructing   half-space and simplicial depths either rely on approximations in dimensions larger than 3 or  computationally demanding  procedures \citep{laketa2023simplicial}.

The Box-CD approach, which is based on ordering   the sample space having a fixed reference $y^{\text{obs}}$, induces   an ordering on the parameter space, similarly to the  idea of the depth-CD of \cite{liu2022datadepth}. The following lemma establishes a connection with the depth concept.

\begin{lemma}\label{lemmaMULTI}
	For two parameter points $\theta^*$ and $\theta^{**}$ within $\Theta$ and with their corresponding random Boxes $\mathcal{B}_t^{*}$, $\mathcal{B}_t^{**}$,
	it holds  
	$$Pr(t^{\text{obs}} \in \mathcal{B}_t^{*})<Pr(t^{\text{obs}} \in \mathcal{B}_t^{**})\Leftrightarrow \mathcal{CD}^{\text{box}}(\theta^*)< \mathcal{CD}^{\text{box}}(\theta^{**}).$$
\end{lemma}
Note that Lemma \ref{lemmaMULTI} is not restricted to the case $p=1$. Indeed,  $\theta$ can be a vector without compromising the definition.
This means that the function $\mathcal{CD}^{\text{box}}(\theta)$ is higher when the random Box $\mathcal{B}_t^{*}$ contains the observed sample often, or equivalently,  that the proposal $(\theta^*, y^{*1}, y^{*2})$ is well centered with respect to the generating process that provided $y^{\text{obs}}$.
Beyond the application of the idea to simulation-based inference, as a Depth function $\mathcal{CD}^{\text{box}}(\theta)$  relies on   hyper-rectangles to determine a centrality measure. This approach reduces the complexity associated with relying on simplexes as in  the SD method.


\subsection{Relation to Confidence sets and properties of the Box-CD function } \label{sect:34}
Remark \ref{remark_diff_conf} plays a crucial role in generalizing the characteristics of the Box-CD function to the multivariate statistical context $(d>1)$. Since  the test statistic $t$ does not need to be stochastically ordered  either in one dimension,  whether the components of $t$ are positively or negatively dependent does not influence the definition of the depth function.

We can generalize Lemma \ref{lemmaconfint} to the general case of confidence sets assuming $\theta \in \mathbb{R}^p$.  Define $M=\max_\theta\mathcal{CD}^{\text{box}}(\theta)$.
\begin{theorem}
	The region $ C_{1-\alpha}=\{\theta | \mathcal{CD}^{\text{box}}(\theta) \geq \alpha M\}$ defines a confidence set with confidence level  $1-\alpha$. 
\end{theorem}
\begin{proof}
	The exact calibration property in the multivariate parameter case follows immediately by the definition.
	In fact, 
	$$Pr_{\theta_0}(\mathcal{CD}^{\text{box}}(\theta)> {Q_{\mathcal{CD}^{\text{box}}}}(\alpha) )=\alpha,$$ indicating that when considering $\mathcal{CD}^{\text{box}}(\theta)$ as a global test statistic, akin to the log-likelihood ratio, $C_{1-\alpha}$ has the nominal coverage property.
\end{proof}

\begin{case}
	Note that, even if the Algorithm \ref{algo0} for deriving the Confidence depth function relies on the definition of hyper-rectangles (in the  space of summary statistics), the confidence regions are curved regions (see Figure \ref{fig:rick}).
\end{case}

\begin{lemma}
	The Box-CD  function   is invariant under any   transformation which is order preserving (up to the sign) applied  individually to the components of $t$.  
\end{lemma}
\begin{proof}
	Define the transformation $w(t): t \in \mathbb{R^d} \mapsto w \in \mathbb{R^d}$ as a bijective, component-wise function, where $w_j(t): t_j \mapsto w_j$ represents the monotonic transformation applied specifically to the $j$-th component of the statistic $t$.
	Then $w_j'> w_j \Leftrightarrow \{t_j'> t_j \text{ or } t_j'< t_j\} $ for any $t_j \in \mathbb{R}. $
	In particualr, given $t_j$ such that  $t_j' < t_j< t_j''$, it follows
	that  $w_j' < w_j< w_j''$ or  $w_j' > w_j> w_j''$ and $Pr( t_j \in \mathcal{B}_t^*|\theta)=Pr( w_j \in \mathcal{B}_w^*|\theta)$.
\end{proof}



\subsection{Efficiency and optimality}
\change {In the Box-CD framework, coverage validity and type-I error control are guaranteed, while the width of the confidence sets reflects the amount of information preserved in the summary statistics. For instance, if a sample of size $n$
	is processed by computing statistics on only a fraction of the available observations, the effective information content decreases, which in turn leads to wider confidence intervals.}

\change{ Under classical regularity conditions—namely independent and identically distributed data from a regular parametric family, smoothness of the likelihood, existence of sufficient statistics (often complete in exponential families), finite Fisher information, and the Monotone Likelihood Ratio (MLR) property—confidence distributions for a scalar parameter   achieve optimality. In particular, they yield confidence intervals with the shortest possible expected length while maintaining the prescribed coverage probability, as their construction aligns with the theory of uniformly most powerful tests. Within this framework, Box-CD–based intervals correspond to equi-tailed intervals derived from the distribution of the pivotal quantity and the confidence distribution is simply the parameter-space representation of the pivot’s distribution, and Box-CD–based intervals emerge as the equi-tailed confidence intervals obtained from this construction.}
\change{
	When constructed from multiple summary statistics, the Box-CD can be directly compared to likelihood ratio methods in terms of efficiency.
	\begin{theorem} \label{thm:opt}
		\textbf{CD-based test versus Likelihood Ratio Test (LRT)}
		Let $y=(y_1,\dots,y_n)$ be i.i.d.\ from a family of distributions 
		$\{p(y \mid \theta): \theta \in \Theta\}$. Suppose that for each $i$, the family $\{p(y_i \mid \theta): \theta \in \Theta\}$ has the 
		{ MLR property} for a given statistic $t(y_i)$.
		Define the likelihood ratio for testing $H_0: \theta = \theta_0$ versus $H_1: \theta \neq \theta_0$ as
		\[
		\lambda(y) =   L(\theta_0 \mid y) /  L(\hat{\theta} \mid y),
		\]
		where $L(\theta \mid y) = \prod_{i=1}^n p(y_i \mid \theta)$ and $\hat{\theta}$ is the MLE under the full parameter space $\Theta$.  The test rejects for small values of $ \lambda(y) $ and  we denote the corresponding critical values at significance level $\alpha$ by $c_\alpha$.
		Then:
		\begin{enumerate}
			\item[(i)] For a single observation $y_i$, the test based on the Box-CD $\mathcal{CD}_i^\mathrm{box}(\theta)$, which rejects $H_0$ when 
			\[
			\mathcal{CD}_i^\mathrm{box}(\theta_0) < \alpha,
			\] 
			is equivalent to the likelihood ratio test (LRT); that is, it has the same rejection region as the LRT applied to $y_i$.
			\item[(ii)] For the full sample $y=(y_1,\dots,y_n)$, the critical region of the 
			LRT at level $\alpha$ is determined by  $c^*_\alpha$, which is the smallest value among the 
			coordinate-wise critical values $c_{i\alpha}$ for individual LRTs:
			\[
			c^*_\alpha \;=\; \min_{1\leq i\leq n} c_{i\alpha}.
			\]
			\item[(iii)] 
			For any $\theta$ and $\alpha$, there exists an integer $k$ such that  when	$\mathcal{CD}^\text{box}(\theta_0)=\alpha$
			$\mathcal{CD}^\text{box}_k(\theta_0)=\alpha_k$, with $\alpha_k \leq \alpha$.
		\end{enumerate}
		As a consequence, the rejection region of the CD-based test contains that of the LRT at the same nominal significance level $\alpha$.  
		Therefore, the power function of the CD-based test satisfies
	 \begin{equation}\label{eq:beta_ineq}
		\beta_\mathrm{CD}(\theta) \;\geq\; \beta_\mathrm{LRT}(\theta), 
		\quad \text{for all } \theta \in \Theta_1,
\end{equation}
		where $\Theta_1$ denotes the parameter values under the alternative hypothesis.  
		This inequality may be strict for some $\theta$, indicating that the CD-based test can achieve strictly higher power than the LRT in these cases.  
		In other words, under the stated MLR and \textcolor{black}{aggregating information from multidimensional statistics}, the CD-based test is uniformly at least as powerful as the LRT and may outperform it for certain alternatives.
	\end{theorem}
	\begin{proof}
		By the MLR property, each marginal Box-CD test $\mathcal{CD}_i^\text{box}(\theta)$ based on $y_i$ is equivalent to the LRT for that observation. In particular, for a single observation,
		\[
		\mathcal{CD}_i^\text{box}(\theta_0) < \alpha_i \quad \Leftrightarrow \quad \lambda(y_i) < c_{i\alpha},
		\]
		so their rejection regions coincide.
		For the full sample $y=(y_1, \dots, y_n)$, let us express the aggregated Box-CD in terms of conditional probabilities. For any two marginals \(i\) and \(j\):
		\[
		\mathcal{CD}^\text{box}(\theta_0) = \frac{\Pr(t_i \in 	\mathcal{B}_i \text{ and } t_j \in 	\mathcal{B}_j)}{M} 
		= \frac{\Pr(t_i \in	\mathcal{B}_i) \, \Pr(t_j \in \mathcal{B}_j \mid t_i \in \mathcal{B}_i)}{M},
		\]
		where \(\mathcal{B}_i \text{ and } \mathcal{B}_j\) are   marginal boxes and \(M\) is the maximum of the unnormalized Box-CD as defined in Section \ref{sect:34}.  
		Since \(\mathcal{CD}_i^\text{box}(\theta_0) = \Pr(t_i \in \mathcal{B}_i)/M_i < \alpha_i\) for some \(i\), and because conditional probabilities satisfy \(0 \leq \Pr(t_j \in \mathcal{B}_j \mid t_i \in \mathcal{B}_i) \leq 1\), and $M_i>M$, the aggregated Box-CD satisfies
		\[
		\mathcal{CD}^\text{box}(\theta_0) \leq \mathcal{CD}_i^\text{box}(\theta_0) < \alpha_i.
		\]  
		Thus, the aggregated rejection region contains all marginal rejection regions, including the LRT rejection region.
		Regarding power, we observe that:
		\begin{itemize}
			\item if marginal statistics are weakly correlated, i.e. \(\Pr(t_j \in \mathcal{B}_j \mid t_i \in \mathcal{B}_i) \approx \Pr(t_j \in \mathcal{B}_j)\), the aggregated Box-CD behaves similarly to the LRT, and power is approximately equal;
			\item if there is moderate conditional dependence, i.e. \(\Pr(t_j \in \mathcal{B}_j \mid t_i \in \mathcal{B}_i) > \Pr(t_j \in \mathcal{B}_j)\), aggregation increases evidence against \(H_0\), so the Box-CD test can achieve strictly higher power than the LRT.
		\end{itemize}
		Hence, for all \(\theta \ne \theta_0\),  Equation \ref{eq:beta_ineq}
	holds,
		with strict inequality under moderate conditional dependence.
	\end{proof}}
\textcolor{black}{\begin{case}
	The conclusion of Theorem 3.7 should not be interpreted as a claim of classical optimality. In particular, the CD-based test is not asserted to be uniformly or locally most powerful in the Neyman-Pearson sense, where such optimality relies on monotone LR  and sufficiency assumptions.
	Rather, the result establishes a containment relationship between rejection regions, implying power equivalence or dominance with respect to the LRT   more in general.
	Importantly, the classical optimality of the LRT does not preclude the existence of alternative procedures that can match its local power, nor does it rule out improved performance in nonregular settings.
\end{case}}
 
\subsection{High dimensional hyper rectangles}

Denote as $\mathcal{B}_t^{*(d)}$ a random Box based on dimension $j \in \{1, \ldots,  d\}$. Then
\change{
	$$Pr(t^{\text{obs}}_1, \ldots, t^{\text{obs}}_{d-1}  \in \mathcal{B}_t^{*(d-1)})\geq Pr(t^{\text{obs}}_1, \ldots, t^{\text{obs}}_{d} \in \mathcal{B}_t^{*(d)}).$$
	This inequality reflects that the acceptance probability decreases as the dimensionality of the statistic increases. Intuitively, for the observed point $t^{\text{obs}}$ to lie within the $d$-dimensional box $\mathcal{B}_t^{*(d)}$, all $d$ components must fall within their respective coordinate-wise intervals. This implies that any subset of $d-1$ components must also fall within their corresponding intervals. In contrast, for the $(d-1)$-dimensional case, only a subset of these constraints needs to be satisfied, making it more likely for the point to be accepted.}

This leads to challenges in accurately estimating the tails of the Box-CD function, as the corresponding parameter regions  are associated with rare events, especially in high dimensions, as it happens for ABC.

\change{However,
as stated in Section 3,  a crucial property of distance and divergence functions typically used in   ABC   is the identity of indiscernible. This condition is not fulfilled by the  type of discrepancy associated to the acceptance criterion of Box-CD. This would only occur in a degenerate case where the set $\mathcal{B}^*_t$ is a singleton with probability one, a scenario not endorsed by our assumptions. The failure to meet this property carries significant implications.  In particular, it allows the method to accept parameter values even when the simulated summary statistics are far from the observed ones, as long as they fall within a coarse acceptance region.  This does not rule out the problem of the curse of dimensionality. But we cannot study this issue straight under the lens of the distance-based approaches.
\subsubsection{ABC for increasing dimension}
To illustrate this, consider this simplified   scenario. Let \( t^{\text{obs}} \in \mathbb{R}^d \) be a fixed observed summary statistic, from the model  $\mathcal{N}_d (0, I_d)$ and let us assume that \( t \sim \mathcal{N}_d (0, V) \) is a simulated statistic from the prior predictive distribution, where  $V=v\cdot I_d$ is diagonal. In ABC, the common acceptance criterion is based on the fixed-radius ball
\[
\| t - t^{\text{obs}} \| \le \epsilon,
\]
for some fixed tolerance \( \epsilon > 0 \). Then the acceptance probability satisfies the identity property.
Now consider without loss of generality $t^{\text{obs}}=0_d$  and $t$  independent random vectors from a standard multivariate normal distribution in \( \mathbb{R}^d \). We are interested in the probability that their Euclidean distance is less than or equal to \( \epsilon \), i.e.
\[
 {Pr}(\|t- t^{\text{obs}}\| \le \epsilon).
\]
Let us denote with $\Delta = t_1 - t_2$ the difference between $t_1$ and $t_2$. Since \( t_1 \) and \( t_2 \) are independent and both distributed as \( \mathcal{N}_d (0, I_d) \), the difference \( \Delta \) is distributed as $\Delta \sim \mathcal{N}_d (0, V+I_d)$.  Therefore, the squared norm follows the  scaled chi-squared distribution
\[
\|\Delta\|^2 \sim (v+1) \cdot \chi^2_d.
\]
Hence, the probability of interest becomes
\[
{Pr}(\|t  - t^{\text{obs}}\| \le \epsilon) = {Pr}((v+1) \cdot \chi^2_d \le \epsilon^2) = F_{\chi^2_d}\left( \frac{\epsilon^2}{v+1} \right),
\]
where \( F_{\chi^2_d} (\cdot) \) is the cumulative distribution function (CDF) of the chi-squared distribution with \( d \) degrees of freedom. The expectation and variance of \( \|\Delta\|^2 \) are, respectively,
\[
\mathbb{E}(\|\Delta\|^2) = (v+1)d \quad \text{and} \quad \text{Var}(\|\Delta\|^2) =  (v+1)^2d.
\]
So, as \( d \to \infty \), the distance   grows roughly as \( \sqrt{(v+1)d} \). For fixed \( \epsilon \), the probability decays rapidly with \( d \). 
Moreover, the density of \( \Delta \sim \mathcal{N}_d (0, (v+1)I_d) \) is approximately constant close to zero (as we are interested in having  $\epsilon\approx 0$). Denoting with $Vol(\epsilon)$ the volume of the ball with radius $\epsilon$, then
\[
{Pr}(\|t - t^{\textbf{obs}}\| \le \epsilon) \approx 
p(t^{\text{obs}}) \cdot Vol(\epsilon) = p(t^{\text{obs}})
\cdot \frac{\pi^{d/2}}{\Gamma(d/2 + 1)} \cdot \epsilon^d,
\]
with $\Gamma\left(\frac{d}{2} + 1\right) \approx \sqrt{\pi d} \left(\frac{d}{2e}\right)^{\frac{d}{2}}$ by   Stirling's approximation.
Hence for large $d$
\[
{Pr}\big(\|t_1 - t_2\| \leq \epsilon\big) \propto \frac{1}{\sqrt{d}} \left( \epsilon \frac{\sqrt{2 \pi e}}{\sqrt{d}} \right)^d  < \epsilon^d,
\]
which decays  super-exponentially for the presence of   \( (1/\sqrt{d})^d \).
\subsubsection{Box-CD for increasing dimension}
Alternatively, consider the acceptance region defined by a data-adaptive axis-aligned hyperrectangle:
\[
\mathcal{B}(t^{(1)}, t^{(2)}) = \{ t \in \mathbb{R}^d : \min(t^{(1)}_j, t^{(2)}_j) \le t_j \le \max(t^{(1)}_j, t^{(2)}_j) \ \text{for all } j = 1, \ldots, d \},
\]
where \( t^{(1)}, t^{(2)} \sim \mathcal{N}_d (0, V) \) are two independent simulated summaries from the prior predictive. Then, the acceptance probability is
\[
{Pr}(t^{\text{obs}} \in \mathcal{B}(t^{(1)}, t^{(2)})) \propto \frac{1}{\prod_{j=1}^d(1-x_j)x_j} \leq \left(\frac{1}{4}\right)^{d},
\] with $0<x<1$ and $x=1/2$ if the prior/proposal is symmetric and centered around the   the point of maximal depth. This quantity decays   exponentially with the dimension, still providing a better rate than ABC.
 }

\subsubsection{Box-CD Extension with $S$ pseudo-samples}
\change{
To alleviate the problem of low acceptance probability that translates into inefficiency in high dimensions, we introduce a generalization of the Box-CD approach based on generating a series of $S$ pseudo-samples instead of a pair, without compromising the validity of the procedure}. We only require that $S$ to be a even number.
Define $$\mathcal{B}_{t,S}^*=\times_{j=1}^{d}[t_{j}^{*
	(1)},t_{j}^{*(S)}],$$ where $t_{j}^{*(1)}$ and $t_{j}^{*(S)}$ are the order statistics along the $j$-th coordinate. Equivalently, the parameter $\theta^*$ is accepted if $t_1^{*(1)} <t_1^{\text{obs}}< t_1^{*(S)} ,\:
t_2^{*(1)}< t_2^{\text{obs}} <t_2^{*(S)},\:   
\ldots 
t_d^{*(1)} < t_d^{\text{obs}} < t_d^{*(S)}.    
$  
The idea is still that of providing a centrality measure but changing the boundaries of the boxes as the minimum and the maximum  test statistics.
The induced ordering relies on the fact that, similarly to Lemma   \ref{lemmaMULTI},
$$Pr(y^{\text{obs}} \in \mathcal{B}_{t,S}^{*})<Pr(y^{\text{obs}} \in \mathcal{B}_{t,S}^{**})\Leftrightarrow \mathcal{CD}^{\text{box}}_S(\theta^*)< \mathcal{CD}^{\text{box}}_S(\theta^{**}).$$ 
\textcolor{black}{The computational cost of a single model simulation, and hence of one acceptance–rejection attempt, is $O(Sd)$ instead of $O(2d)$.
 However, the number of accepted samples may increase at a rate exceeding the ratio $S/2$, which can lead to more accurate estimation of the target function—especially in the tails—under a fixed computational budget. Moreover, increasing $S$ naturally lends itself to parallel computation: with $S$ parallel processors, the acceptance probability can grow by more than a factor of $S$ in practice. Finally, note that in general   the cost of simulating from the model is typically  higher than that of computing the statistics, thus $S$ dominates the cost for each proposal from the model.}

In Example 4.3 we empirically examine the effect of choosing   $S>2$, with particular attention to scenarios where the dimension of  summary statistics increases.


\section{Examples} We present and discuss a series of examples across both classical problems as GLMs, and more challenging cases from the domain of LFI.
For each example considered, we perform  a simulation study with 2000 replicated datasets for each model to assess the validity of the coverage of confidence sets provided by the proposed method and to compare the results with those provided by the Likelihood Ratio test, when available. The results of these simulation experiments are reported in   Table     \ref{tab:sim};  Monte Carlo standard errors for the empirical  coverage are  between 0.008 and  0.01. 
The code for reproducing all the simulations is available at \url{https://github.com/elenabortolato/box}.
In all the given scenarios, after executing Algorithm \ref{algo0}, we perform density estimation using independent Gaussian kernels, as implemented  in the  R library \texttt{pdfCluster} \citep{azzalini2014clustering}. Specifically, the density estimate for a point $\mathbf{x} \in \mathbb{R}^d$ is given by
\[
\hat{f}(\mathbf{x}) = \frac{1}{n h^d} \sum_{i=1}^n K\left( \frac{\mathbf{x} - \mathbf{x}_i}{h} \right),
\]
where $h > 0$ is the bandwidth parameter, internally chosen via cross-validation, and
$K(\mathbf{u}) = \prod_{j=1}^d k(u_j)$, with $k(u_j)$ representing the univariate Gaussian kernel function independently applied across each dimension $j$. We employed a 1-Nearest neighbor method to assess whether   $\theta_0$ was included in the confidence regions, by predicting the value of $\mathcal{CD}^\text{box}(\theta_0)$.
\textcolor{black}{The reported computing times in each example correspond to computations performed on a 3.49 GHz processor and parallelizing on 7 cores.}

\subsection{Logistic regression}
Consider a logistic regression model for a sample   of size $n=20$ with $p=3$ predictors and corresponding coefficients equal to $\beta_0=( -0.25,0,0.25)$.  The summary statistics employed comprise the model's sufficient statistics $t=X^\top y$,  of dimension $d=p$.
As a proposal, we use $\pi(\beta)=\text{Uniform}[-6,6]^p$. 
The empirical  coverage level of confidence sets  are closer to their nominal value than those obtained via the Likelihood Ratio test (Table \ref{tab:sim}).  \textcolor{black}{  The computing time for $R=10000$  replications  is 1.80 seconds.}


\begin{table}[h!]
	\centering
	\begin{tabular}{p{0.068cm}p{0.068cm}p{0.068cm}p{0.08cm}p{1.cm}|cccc|cccc}
		& & & & &\multicolumn{4}{|c|}{\textbf{CD-Box $(1-\alpha)$}} & \multicolumn{4}{|c}{\textbf{ LR $(1-\alpha)$}} \\  
		$p$ &$d$ &$S$& $n$&  Model& \textbf{0.95} & \textbf{0.90} & \textbf{0.85}&\textbf{0.8}&   \textbf{0.95} & \textbf{0.90} & \textbf{0.85}&\textbf{0.8}\\
		\hline
			3&3&2& 20 &Logistic    & 0.948 & 0.899 &0.850 &0.787 & 0.929& 0.868 &0.811 &0.758 \\
				3&3& 2& 10 & M$t$       & 0.956 &0.900& 0.851 &0.780   &0.942 &0.884& 0.828 &0.771\\
		1&10&6& 10&  Mixture   & 0.959 & 0.897  & 0.823 & 0.778& 0.949 & 0.893 &0.840 &0.789    \\
		1&3&4&50&Ricker's &0.936 & 0.890&0.848& 0.758 &&&&\\
		2&19&10&20&Ricker's &0.938 & 0.886 & 0.842 & 0.794 &&&&
	\end{tabular}
	\caption{Results from the simulation studies based on 2000 replicated datasets for each model.  {Left}: coverages from the proposed method (CD-Box).  {Right}: coverages with the Likelihood Ratio (LR) test (when available).}
	\label{tab:sim}
\end{table}

\subsection{Multivariate $t$ distribution} Consider a three-variate Student's $t$ model with $10$ degrees of freedom for $n=10$ observations,  unknown vector of non-centrality parameter $\mu$ and known covariance matrix  given by
$$\Sigma=\begin{pmatrix}
	2&-1&0.4\\
	-1&1.6&0.7\\
	0.4& 0.7&1
\end{pmatrix}.$$ The true data generating parameter was set to $\mu_0=(0,-0.5,0.5)$.
As a  proposal we use $\text{Uniform}[-5,5]^3$ and as   summary statistics the empirical medians of the components. The number of pseudo-samples generated for each parameter proposal was $S=2 $. The results in Table \ref{tab:sim} show that the method guarantees nominal coverage for confidence regions. \textcolor{black}{The computing time in this example for $R=10000$ replications is of 2.17 seconds}.

\change{In this example the summary statistics, in addition to being dependent, have a distribution that depends on  all the unknown parameters.
	 By contrast, if the scale parameters were unknown but the summaries included only the means, the acceptance probability would not reach the upper bound of $1$ for increasing scales, violating the assumptions.  If the correlation parameters were unknown instead,  by using only the empirical means as summaries, the acceptance probability would be constant marginally in the correlations -  again violating the assumptions and leading to  non-informative regions whose depth is constant.}

\subsection{Mixture}
Consider the normal mixture model $y\sim 0.5 \mathcal{N}(-\theta,1)+0.5 \mathcal{N}(\theta,1).$ The summary statistic used is the ordered sample $t=(y^{\text{obs}}_{(1)}, \ldots ,y^{\text{obs}}_{(n)})$, of size $d=n$. The proposal for $\theta$ is a $\text{Uniform} 
[0,3]$ and we fix  $n=10,15,20,25$.  \textcolor{black}{ The computing time for $R=10000$ replications with $n=10$ is 2.5 seconds.}

We \change{focus on} this model to  study the acceptance ratio as a function of the dimension of the summary statistics ($d$) and the number of pseudo-samples, governed by the  hyper-parameter $S$. Figure \ref{fig:accmix} reports the \change{total}  number of accepted proposals from draws of size $R=100000$
(left) and ratio compared to accepted proposals with $S=2$. For a fixed $S$,  the number of accepted draws   decreases as $d$ increases, causing loss in efficiency. When $S$ varies, the number of accepted parameters is not proportional to $R$, but grows faster (see the right panel of Figure \ref{fig:accmix}).

\begin{figure}[h!]
	\centering
	\includegraphics[width=0.98\linewidth]{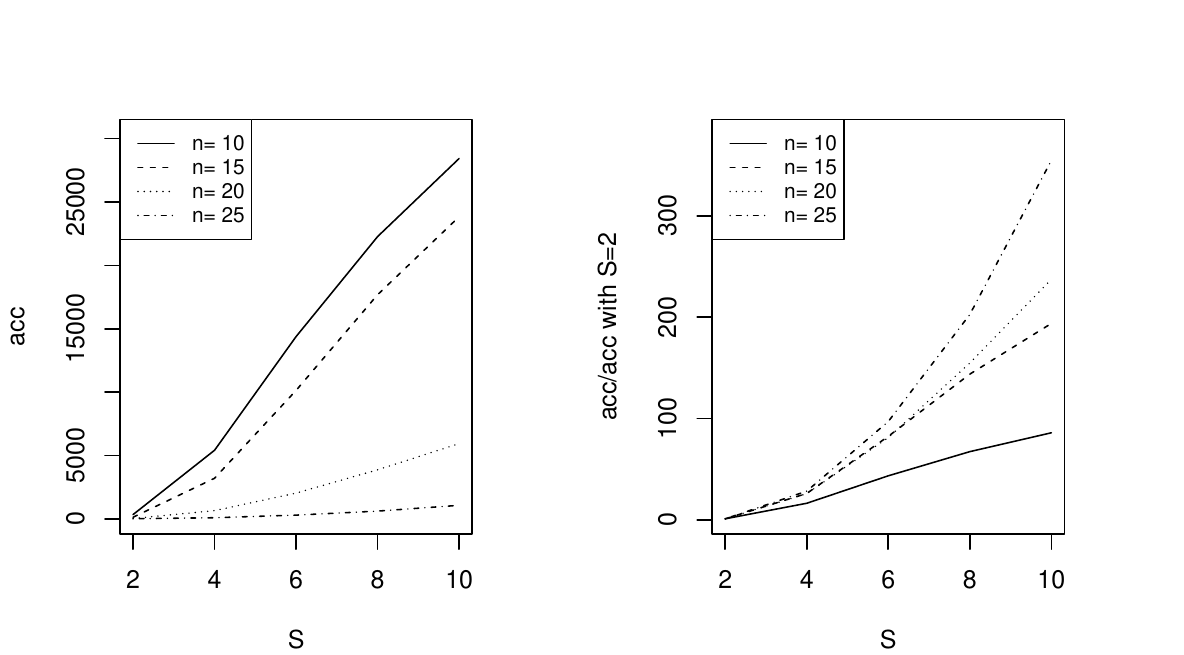}
	\caption{\textit{Left}: number of accepted parameters  from 100000 proposals from the Mixture example, for varying sample size $n$ and number of pseudo-samples $S$. \textit{Right}: ratio of accepted parameters  to accepted with $S=2$.   }
	\label{fig:accmix}
\end{figure}

\begin{figure}[h!]
	\centering
		\includegraphics[width=1\linewidth]{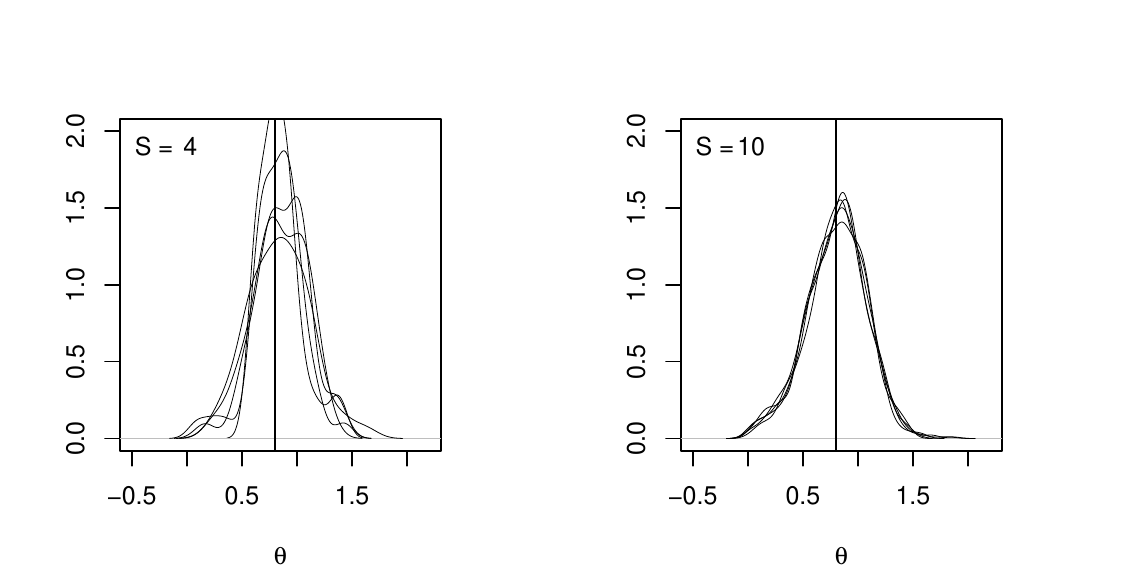}
	\caption{Five replications of the same Box-CD function, with fixed $y^{\text{obs}}$ for the position parameter in the mixture model,  with number of pseudo-samples $S$  varying. The vertical line indicates the true generating parameter $\theta_0=0.8$.}
	\label{fig:mixvar}
\end{figure}

Figure \ref{fig:mixvar} presents the Box-CD functions derived from a sample of size $n=25$ drawn using as a true generating parameter $\theta_0=0.8$. \change{  For every value of $S=4$ and $S=10$, five replications of the Box-CD are generated using the same observed sample.} When $S$ grows, the procedure's variability diminishes. Note that the tails of the function become heavier as $S$ grows. This phenomenon poses no problem as demonstrated by the simulation study (Table \ref{tab:sim}) because the confidence sets are reliant on the value of the function $CD(\theta)$ instead of tail areas. In conducting the simulation study for assessing the coverage properties of the resulting confidence intervals, we considered  $n=d=10$ and  set $S=6$.  \change{   In this example, the average lengths of the   0.95,  0.90 , 0.85 and 0.80 confidence intervals based on the LRT were 2.56, 2.88, 3.08 and  3.31, respectively, while those based on the CD-Box   were 1.22, 1.37, 1.55 and 1.84, respectively.}

\subsection{Ricker's Model}
Consider the Ricker's model \citep{ricker1954}, 
which describes the evolution of the number of animals of a certain species by
$$\log(N(t))= \log(r) + \log(N(t-1)) - N(t-1) +\sigma e(t),$$ where $N(t)$ is the unknown population at time $t$, $\log(r)$ is the logarithmic growth rate, $\sigma$ is the standard deviation of innovation and $e(t)$ is an independent Gaussian error. Given $N(t)$, the observed population at time $t$  is a Poisson random variable, 
$y_t\sim \text{Poisson}(\phi N(t)) $, where $\phi$ is a scale parameter.  
The likelihood for this model is intractable.

We conduct two experiments: first, we assume that only the log-growth rate is unknown
and consider as summary statistics
the median of counts and the quantiles of level 0.25\change{,} 0.75. For the second experiment, both the parameters $log(r)$ and $\sigma^2$ were considered unknown, and we used as the set of summary statistics the whole time series minus the first observation, thus of length $d=19$. 
The number of pseudo-samples for each proposals were $S=2$ and $S=10$ in the two experiments, respectively. The empirical coverages of the confidence sets  are conformal with the nominal (Table \ref{tab:sim}). Two examples of confidence regions obtained for two independent draws from the model with parameters $log(r)=2$ and $\sigma^2=2$ are reported in Figure \ref{fig:rick}. \textcolor{black}{The computing time for $R=10000$ replications in the forst experiment is 2.28 seconds, for the second experiment 4.52 seconds.}
\begin{figure}
	\centering
	\includegraphics[width=0.45\linewidth]{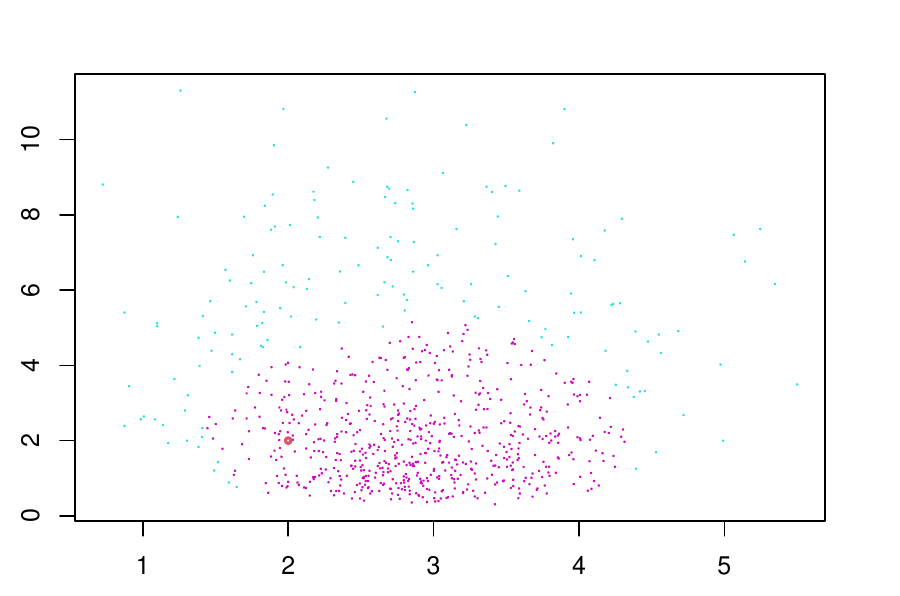}
	\includegraphics[width=0.45\linewidth]{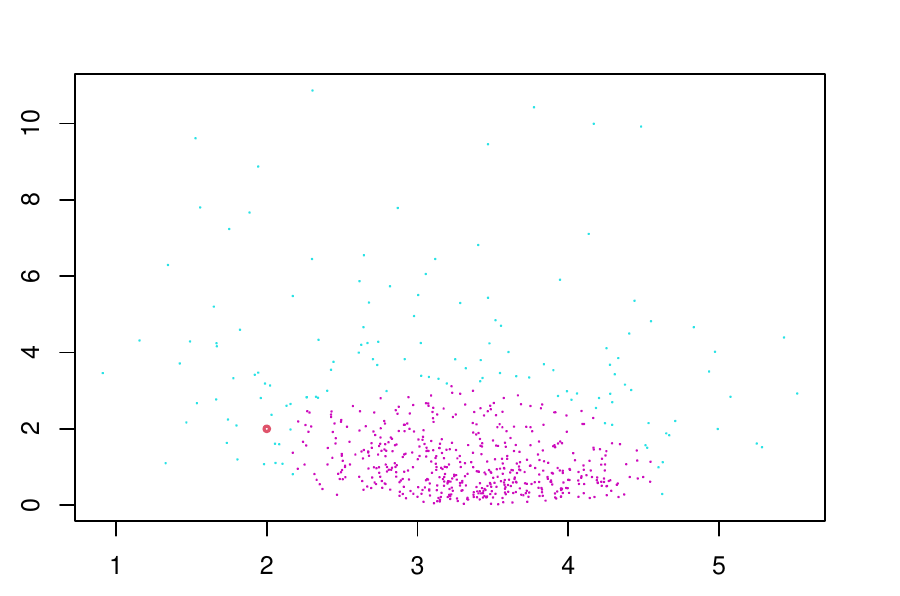}
	\caption{Two Monte-Carlo confidence regions for the parameters $\log(r)$ and $\sigma^2$ in the Ricker's model, the firts ({left}) containing the true generating parameter, the second ({right}) failing in including the paramer. }
	\label{fig:rick}
\end{figure}

\section{Discussion}
The Box Confidence Depth algorithm  introduced in this paper provides a simple yet effective method  to construct calibrated confidence intervals and regions in both likelihood-based and likelihood-free scenarios, making it versatile across various statistical contexts.
The method is designed to work with multivariate parameters and potentially multivariate test statistics; in fact, it effectively uses a measure of centrality of observed data with respect to simulated data, providing intuitive ordering in multivariate spaces.

There are several areas for potential improvement. As with many Monte Carlo  methods, the procedure may be demanding in terms of computational resources,  especially for high-dimensional problems. To boost the computational efficiency of the method, techniques such as adaptive proposals, resampling strategies, and  methods for simulating rare events may be utilized 
\citep{ tokdar2010importance,
	caron2014some, bugallo2017adaptive}.
Automated methods for selecting optimal summary statistics, even multivariate, in the absence of domain knowledge could enhance the method's applicability. In particular, Machine Learning methods, and contrastive learning  approaches can be used to learn summary statistics \citep[see][]{fearnhead:prangle:2012, cranmer2015approximating, jiang2017learning, wang2022approximate}. Similarly, advanced methods for the essential density estimation step, such as Normalizing Flows \citep{kobyzev2020normalizing} could be adapted. A detailed exploration of these methods in this context   presents an interesting direction for future research.

\begin{funding}
The first author acknowledges funding from the European Union under the ERC grant project number 864863 and from the Severo Ochoa Programme for Centres of Excellence in R\&D (Barcelona School of Economics CEX2024-001476-S), funded by MCIN/AEI/10.13039/501100011033.
\end{funding}

\bibliographystyle{imsart-nameyear.bst} 

\bibliography{paper-ref}

\end{document}